\documentclass[10pt, a4paper]{amsart}
\usepackage{amssymb,amsmath,latexsym,amscd,amsfonts,amsthm,mathrsfs}
\usepackage{enumerate, graphicx}
\usepackage[utf8]{inputenc}

\newcommand{\re}{\text{\rm Re\,}}
\newcommand{\im}{\text{\rm Im\,}}
\newcommand{\ind}{\text{\rm ind\,}}
\newcommand{\kker}{{\textrm{Ker\,}}}
\newcommand{\eeq}{{\, \stackrel\centerdot=\, }}
\newcommand{\lee}{{\, \stackrel\centerdot\le\, }}
\newcommand{\gee}{{\, \stackrel\centerdot\ge\, }}
\newcommand{\gse}{{\, \stackrel\centerdot>\, }}
\newcommand{\lse}{{\, \stackrel\centerdot<\, }}
\newcommand{\opp}[1]{{({#1})_\oplus}\, }

\newcommand{\ie}{{\emph{i.e.},\ }}
\newcommand{\bi}{{\Big |}}
\newcommand{\dpp}{{^\prime{}^\prime}}

\newcommand{\bn}{{\mathbb{N}}}
\newcommand{\br}{{\mathbb{R}}}

\newcommand{\bc}{{\mathbb{C}}}

\newcommand{\ca}{{\mathcal{A}}}

\newcommand{\cdd}{{\mathcal{D}}}
\newcommand{\cb}{{\mathcal{B}}}
\newcommand{\cl}{{\mathscr{L}}}
\newcommand{\cf}{{\mathscr{F}}}
\newcommand{\ce}{{\mathscr{E}}}
\newcommand{\cg}{{\mathscr{G}}}
\newcommand{\cn}{{\mathcal{N}}}
\newcommand{\cm}{{\mathcal{M}}}
\newcommand{\bff}{{\mathcal{B}_{\mathscr{F}}}}

\newcommand{\ld}{{\ell^2_d}}
\newcommand{\ldo}{{\ell^2_{d, o}}}
\newcommand{\lde}{{\ell^2_{d, e}}}

\renewcommand{\a}{\alpha}
\renewcommand{\b}{\beta}
\renewcommand{\l}{\lambda}
\newcommand{\s}{\sigma}

\renewcommand{\d}{\delta}

\newcommand{\ep}{\varepsilon}
\newcommand{\epp}{\epsilon}

\newcommand{\pp}{\Phi}
\newcommand{\pps}{\Psi}

\newcommand{\nt}{\noindent}

\newcommand{\ovl}{\overline}

\newcommand{\ti}{\tilde}
\newcommand{\lb}{\left[}
\newcommand{\rb}{\right]}
\newcommand{\lp}{\left(}
\newcommand{\rp}{\right)}
\newcommand{\lt}{\left}
\newcommand{\rt}{\right}

\newcommand{\la}{\langle}
\newcommand{\ra}{\rangle}
\newcommand{\nn}{\nonumber}

\newcommand{\tr}{{\textrm{tr \,}}}

\newcommand{\rre}{{\rm{Re\, }}}

\allowdisplaybreaks
\numberwithin{equation}{section}

\newtheorem{theorem}{Theorem}[section]
\newtheorem*{theorema}{Theorem A}
\newtheorem*{theoremb}{Theorem B}
\newtheorem{lemma}[theorem]{Lemma}

\newtheorem{proposition}[theorem]{Proposition}

\begin{document}

\title[Instability of the essential spectrum for Jacobi matrices]
{On the instability of the essential\\ spectrum for block Jacobi matrices}

\author{S. Kupin}
\address{IMB, CNRS, Universit\'e de Bordeaux, 351 ave. de la Lib\'eration, 33405 Talence Cedex, France}
\email{skupin@math.u-bordeaux.fr}

\author{S. Naboko}
\address{Physics Institute, St. Petersburg State University, Ulyanovskaya str. 1, St. Peterhof, St. Petersburg, 198-904 Russia}
\email{sergey.naboko@gmail.com}


\begin{abstract}
We are interested in the phenomenon of the essential spectrum instability for a class of unbounded (block) Jacobi matrices. We give a series of sufficient conditions for the matrices from certain classes to have a discrete spectrum on a half-axis of a real line. An extensive list of examples showing the sharpness of obtained results is provided.
\end{abstract}

\maketitle

\begin{center}
{\it To Leonid Golinskii on the occasion of his 65-th anniversary}
\end{center}

\section*{Introduction}\label{s0}

Being given a self-adjoint operator on a Hilbert space, the spectral structure of its “relatively small” (\ie relatively compact, relatively trace class, etc.) perturbations is nowdays well understood, see Kato \cite{ka}, Reed-Simon \cite{resi1}, and Birman-Solomyak \cite{birso1}.
For instance, accordingly to Weyl theorem,  any compact perturbation of a bounded self-adjoint Jacobi matrix preserves its essential spectrum, see Kato \cite[Thm. 5.35]{ka}. When the perturbation is of trace class, the absolutely continuous spectrum of the Jacobi matrix is also preserved, see Kato \cite[Thm. 4.4]{ka}.

The main subject of this article are unbounded (block) Jacobi matrices. The peculiarity of this class of operators is that one is often interested in spectral properties of a perturbation of a model Jacobi matrix, and the perturbation is not “relatively small” in the sense of the previous paragraph. Still, it is frequently “reasonably small” from the point of view of various applications. To illustrate this, consider Jacobi matrices
$$
J_\a=J(\{n^\a\},\{0\}),\quad J_{\a,\b}=J(\{n^\a\},\{n^\b\}),
$$
where $0<\a\le 1,\ 0<\b<\a$, see \eqref{e101-} for definitions. Despite the fact that the diagonal elements of $J_{\a,\b}$ are small, or even negligible, as compared to off-diagonal elements on the power scale, it is easy to see that $J_{\a,\b}$ is neither relatively compact nor relatively bounded perturbation of $J_\a$, see Subsection \ref{s301} for a sketch of a proof.

Hence, the study of the spectral problems of the above type leads to results on different types of \emph{spectral instability} or \emph{spectral phase transition}, see Damanik-Naboko \cite{dana}, Janas-Naboko \cite{jana1}-\cite{jana4}, Janas-Naboko-Stolz \cite{jana6}, Naboko-Pchelintseva-Silva \cite{napche}, Naboko-Simonov \cite{nasi}. 

In the present article, we study the problem of the essential spectrum in\-sta\-bi\-lity with respect to perturbations of  unbounded block Jacobi matrices, including the perturbations with slowly growing entries as compared to the entries of the unperturbed Jacobi matrix. Leaving aside the above cases of “relatively small”  perturbations of a Jacobi matrix,  we prove that “straightforward'' counterparts of  classical results  are not valid anymore in this situation.

We introduce some notation to formulate our theorems. A block Jacobi matrix (or, a Jacobi matrix with matrix entries) $J: \ell^2(\bn; \bc^d)\to \ell^2(\bn; \bc^d), \ d\ge 1$, is defined as 
\begin{equation}\label{e1}
J=J(\{A_n\}, \{B_n\})=
\begin{bmatrix}
B_1& A_1& 0& \ldots\\
A^*_1& B_2& A_2& \ddots\\
0& A^*_2&B_3& \ddots\\
\vdots& \ddots&\ddots&\ddots
\end{bmatrix},
\end{equation}
where $A_j, B_j\in \mathcal{M}_{d,d}(\bc)$ are $d\times d$ complex matrices. We assume that $B_j=B^*_j$ and $A_j$ is invertible for all $j$.  The zero and the identity operators on $\bc^d$ are denoted by $0_d$ and $I_d$, respectively. When $d=1$, we say that $J$ is a scalar-valued Jacobi matrix, or just a Jacobi matrix to be short, \emph{i.e.},
\begin{equation}\label{e101-}
J=J(\{a_n\}, \{b_n\})=
\begin{bmatrix}
b_1& a_1& 0& \ldots\\
a_1& b_2& a_2& \ddots\\
0& a_2&b_3& \ddots\\
\vdots& \ddots&\ddots&\ddots
\end{bmatrix},
\end{equation}
where $a_j,b_j\in\br$ and $a_j\not=0$ for all $j$. 

Let $\{B_n\}\subset \cm_{d,d}(\bc)$ be a sequence of Hermitian matrices. We say that
\begin{equation}\label{e102-}
\lim_{n\to\infty} B_n=+\infty,
\end{equation}
if  for any $M\ge 0$ there is a $N=N(M)$ such that $B_n\ge M I_d$ for $n\ge N$.
\begin{theorema}[{= a short version of Theorem \ref{t1}}]\label{t001} Let $J$ be a block Jacobi matrix \eqref{e1}. Suppose that:
\begin{enumerate} 
\item for some fixed $c\in\br$, there is a number $N_0$ such that $B_{2n}\le c I_d$ for $n\ge N_0$.  
\item we have
$$
\lim_{n\to\infty} B_{2n-1}=+\infty. 
$$
\end{enumerate}
Then the spectrum of $J$ (denoting an arbitrary self-adjoint extension of $J_{min}$, see \eqref{e1016}) is discrete in $(c,+\infty)$, and it accumulates to $+\infty$ only.
\end{theorema} 
The proof of the above theorem uses, besides all, an elementary spectral variational principle applied to an appropriate increasing family of subspaces,  see Weid\-mann \cite[Satz 8.28, 8.29]{we}. 

Several remarks are in order. First, the above theorem trivially holds true if one permutes the rôles of $\{B_{2n}\}$ and $\{B_{2n-1}\}$. Second, the theorem \emph{does not depend} on the behavior of the sequence $\{A_n\}$, and, third, it imposes rather mild conditions on the growth of the diagonal subsequence $\{B_{2n-1}\}$. In particular, there is no condition on the rate of growth of matrix entries $B_{2n-1}$ to $+\infty$. It follows that an arbitrary small growth of the entries of a diagonal perturbation $\{B_{2n-1}\}$ can change dramatically the essential spectrum of the matrix. For instance,  an unbounded self-adjoint Jacobi matrix $J(\{A_n\},\{0_d\})$ can have the essential spectrum filling in the whole real line. At the same time, its diagonal self-adjoint perturbation $J(\{A_n\}, \{B_n\})$, with blocks $B_{2n-1}$ growing arbitrarily slowly at infinity, can change the essential spectrum on a half-line to the purely discrete spectrum. In particular, we see that even the absolutely continuous spectrum is completely unstable  with respect to diagonal perturbations of the above type; see examples given in Section \ref{s3} for more details.

The next theorem addresses the same effect, but it allows one to have a certain interaction between diagonal and off-diagonal components of the block Jacobi matrix.
 For a $B\in\cm_{d,d}(\bc), B^*=B$, let
\begin{equation}\label{e1031}
\opp{B^{-1}}=
\lt\{
\begin{array}{lcl}
(B_+)^{-1}&, & \ \mathrm{on} \ \im (B_+),\\
0&,&  \ \mathrm{on} \ \im (B_+)^\perp,
\end{array}
\rt.
\end{equation}
where $B_+$ stands for the positive part of $B$, see the discussion before Theorem \ref{t21}.

\begin{theoremb}[{= Theorem \ref{t22}}]\label{t002} Let $J$ be the Jacobi matrix defined by \eqref{e1}. Assume that:
\begin{enumerate}
\item $\lim_{k\to\infty}B_{2k-1}=+\infty$,
\item $\lim_{k\to+\infty} \opp{B^{-1}_{2k}}=0$, 
\item moreover, one has 
\begin{equation*}
\limsup_{n\to+\infty} ||\opp{B^{-1/2}_{2n}} A^*_{2n-1} B^{-1/2}_{2n-1}||+
\limsup_{n\to+\infty}||\opp{B^{-1/2}_{2n-2}} A_{2n-2}B^{-1/2}_{2n-1}||<1.
\end{equation*}
\end{enumerate}
Then the part of the spectrum $\s(J)\cap(0, +\infty)$ is discrete.
\end{theoremb}
In other words, condition (2) of the theorem says that either 
$$
\lim_{k\to\infty}(B_{2k})_+\Big|_{\im((B_{2k})_+)}=+\infty,
$$ 
or $(B_{2k})_+=0$. Among other techniques, the proof of this theorem uses a generalized version of Schur-Frobenius lemma, see Subsection \ref{s14}. 

Theorems A and B are new even in the case of scalar-valued Jacobi matrices. Note also that an explicit special case of the phenomenon described in the above theorems, was studied in Damanik-Naboko \cite{dana}, see Example 1, Subsection \ref{s31}. The class of Jacobi matrices from \cite{dana} illustrated  the so-called spectral phase transition phenomenon of the second kind, see also Janas-Naboko \cite{jana1}-\cite{jana3}.  Hence, one can consider the present theorems  as generalizations of results from the above mentioned articles.

The paper is organized as follows. In Section \ref{s1}, we introduce the notation and give a brief list of facts on Jacobi and block Jacobi matrices we shall use in the paper. It contains some basic facts on self-adjoint Jacobi matrices, Gilbert-Pearson subordinacy theory in the Jacobi matrix case \cite{gipe, khape}, Levinson's type asymptotic results \cite{jamo1, jamo2}, and a special version of classical Schur-Frobenius lemma \cite[Sect. 1.6]{tre}. Theorem A along with its corollaries is proved in Section \ref{s01}. Theorem B is given in Section \ref{s2}. Section \ref{s3} presents a series of examples illustrating the sharpness of obtained results as well as counter-examples to certain attempts of their generalizations.

Concluding the introduction, we say a few more words on the notation. For a separable Hilbert space $H$, the zero and the identity operators are denoted by $0_H$ and $I_H$, respectively. As explained above, for $H=\bc^d$, we write $0_d$ and $I_d$, correspondingly. When the spaces we work on is clear from the context the subindex $(.)_H$ is dropped. The writing $\cm_{d,d}(\bc)$ stays for the algebra of $d\times d$ complex matrices.

\section{Preliminaries}\label{s1}

\subsection{Notation and generalities on (block) Jacobi matrices}\label{s11}
Recall the definitions of a block  and a scalar-valued Jacobi matrices given in \eqref{e1} and \eqref{e101-}, respectively. In general, we prove our results for block Jacobi matrices, specializing, if needed, to the scalar-valued case.

To keep the notation simple, the space $\ell^2(\bn; \bc^d)$ is denoted $\ld$; of course, we write $\ell^2$ for $\ell^2_1$. Sometimes, we have to indicate precisely the space $\bc^d$, corresponding to the $j$-th ``component'' of a vector $u=\{u_j\}\in\ld$. We shall write it as $(\bc^d)_j$ so that $u_j\in (\bc^d)_j$. We set $P_j:\ld\to (\bc^d)_j$ to be the corresponding orthoprojector, \emph{i.e.}, $P_ju=u_j\in (\bc^d)_j$.

We shall use also orthoprojectors on ``odd'' and ``even'' subspaces of $\ld$. That is, for a 
$u=\{u_j\}\in\ld$, we represent it as a sum of vectors $u_o, u_e$ defined as
\begin{eqnarray*}
u_o&=&P_o u=
\lt\{
\begin{array}{lcl}
0_d&,& j=2k,\\
u_j&,& j=2k-1,
\end{array}\rt.\\
&&\\
u_e&=&P_e u=
\lt\{
\begin{array}{lcl}
u_j&,& j=2k,\\
0_d&,& j=2k-1,
\end{array}\rt.
\end{eqnarray*}
where $k\in\bn=\{1,2,\dots\}$.
Furthermore,
\begin{eqnarray*}
\ldo&=&\{v=\{v_j\}_{j\in\bn}\in \ell^2_d: v_{2k}=0_d, \ k\in\bn\},\\
\lde&=&\{v=\{v_j\}_{j\in\bn}\in \ell^2_d: v_{2k-1}=0_d, \ k\in\bn\}.
\end{eqnarray*}

We put $P_e:\ell^2_d\to\lde$ and $P_o:\ell^2_d\to\ldo$ to be the orthoprojectors on the above subspaces.  For $u\in\ld$,  we have trivially $u=u_o\oplus u_e$ with $u_o\in\ldo$ and $u_e\in\lde$.

We  now remind some basic facts on self-adjoint extensions of Jacobi matrices \eqref{e1}. Self-adjoint extensions of scalar-valued Jacobi matrices are discussed in Teschl \cite[Sect. 2.6]{te1}. The case of block Jacobi matrices is in Berezanskii \cite[Sect. VII.2]{ber1}, see also Damanik-Pushnitskii-Simon \cite{dapu} in this connection. For the general theory of self-adjoint operators, see Reed-Simon \cite{resi1}, Birman-Solomyak \cite{birso1} and Weidmann \cite{we1}.

For a given $J$ \eqref{e1}, define its minimal and maximal domains
\begin{eqnarray}\label{e1016}
\cdd_{min}&=&\cdd(J_{min})=\ell^2_{d, \emptyset}=\{u=\{u_j\}\in\ld: \ u_j=0, j\ge N_u\},\\
\cdd_{max}&=&\cdd(J_{max})=\{u=\{u_j\}\in\ld: Ju\in \ld\}.
\end{eqnarray}
It is easy to see that $J_{min}$ defined on $\cdd_{min}$ is a symmetric operator, and its adjoint is exactly $J_{max}$ defined on $\cdd_{max}$, see \cite[Ch. VII]{ber1}. One can show that the defect numbers satisfy the inequality
$$
d_\pm(J)=\dim\kker (J_{max}\mp zI)\le d,
$$
where $z\in \{z: \im z>0\}$.  We assume \emph{for the rest of the paper} that the defect numbers $d_\pm(J)$ are equal, \ie $J$ admits a family of self-adjoint extensions. If this family consists  of a single element, then $J_{max}$ is itself self-adjoint,\emph{i.e.}, $d_\pm(J_{max})=0$. Then we say that $J$ is in  “limit point case”, see \cite{ber1, te1} for more details. 

It is well-known that the distributional characteristics of discrete spectra of two distinct self-adjoint extensions of $J$ are the same \cite[Ch. 3, 4]{birso1}. Consequently, our results on the distribution of discrete spectrum of $J$ do not depend on the choice of a particular self-adjoint extension, and we keep it fixed for the rest of the paper. For the sake of simplicity, this self-adjoint extension of $J$ will be denoted by $J$ as well.

\subsection{On Gilbert-Pearson subordinacy theory for Jacobi matrices}\label{s12}
We recall certain basic facts from subordinacy theory for scalar-valued Jacobi matrices, see Khan-Pearson \cite{khape} and the seminal paper Gilbert-Pearson \cite{gipe}.

Let $J=J(\{a_n\},\{b_n\})$ be a Jacobi matrix \eqref{e101-} in the limit point case. Consider the space of generalized eigenvectors corresponding to $\l\in\bc$, \emph{i.e.}, 
\begin{equation}\label{e1012}
(Ju)’(\l)=\l u’(\l), 
\end{equation}
where $u(\l)=\{u_n(\l)\}$ and $u’(\l)=\{u_n(\l)\}_{n>1}$. Note that the vector $u’(\l)$ is not required to lie in $\ell^2$. One says that a solution $u_s(\l)=\{u_{s,n}(\l)\}$ to the above problem is subordinate, if
$$
\lim_{N\to+\infty} \frac{\sum^N_{n=1}|u_{s,n}(\l)|^2}{\sum^N_{n=1}|v_n(\l)|^2}=0
$$
for any other solution $v(\l)=\{v_n(\l)\}$ to the problem, which is linearly independent of $u_s$.

By spectral theorem, we associate the (scalar-valued) standard spectral measure $\mu=\mu(J)$  to $J$. We denote by $\mu_{ac}, \mu_{s}, \mu_{sc}$ and $\mu_d$ the absolutely continuous, singular, singular continuous, and discrete components of $\mu$, respectively. Let $\cm_{ac},\cm_s,\cm_{sc}$ and $\cm_d$ be the minimal supports of these measures.

The following theorem will be repeatedly used in Section \ref{s3}.
\begin{theorem}[{Khan-Pearson \cite[Thm. 3 ]{khape}}]\label{t01}
Let $J$ be the above Jacobi matrix. The minimal supports  $\cm_{ac},\cm_s,\cm_{sc}$, and $\cm_p$ of $\mu_{ac}, \mu_{s}, \mu_{sc}$, and  $\mu_d$ are as follows:
\begin{itemize}
\item $\cm_{ac}=\{ x\in\br:$ no subordinate solution to $(Ju)’=xu’$ exists $\}$,
\item $\cm_s=\{ x\in\br:$ a subordinate solution to $(Ju)’=xu’$ exists $\}$, 
\item $\cm_{sc}=\{ x\in\br:$ a subordinate solution to $(Ju)’=xu’$ exists, but $u’\not\in\ell^2$ $\}$,
\item $\cm_d=\{ x\in\br:$ a subordinate solution to $Ju=xu$ exists, it satisfies boundary condition $b_1u_1+a_1u_2=x u_1$, and $u\in\ell^2$ $\}$.
\end{itemize}
\end{theorem}

\subsection{Some facts from Levinson theory on Jacobi matrices}\label{s13}
In Section \ref{s3}, we shall have to compute asymptotics of eigenvector equation \eqref{e1012} for some special Jacobi matrices. An appropriate framework for doing this is the so-called Levinson theory, see Coddington-Levinson \cite{cole} for the case of differential operators and Benzaid-Lutz \cite{belu}, Elaydi  \cite{elay} for the case of finite differences.
 
More concretely, we shall use the following result from Janas-Moszi\'nski \cite{jamo1}. Let $n_0\in\bn$, and $\{A_n\}_{n\ge n_0}$ be a uniformly bounded sequence of elements from $\cm_{d,d}(\bc)$, \emph{i.e.,}  for some $M\ge 0$, we have $||A_n||\le M$  for $n\ge n_0$. Consider vectors
 $x^j=\{x^j_k\}_{k\ge n_0}\in \ld$, where $x^j_{n_0}=e^j, \ j=1,\dots, d$, and $\{e^j\}$ is the standard basis of $\bc^d$. Suppose that
\begin{equation}\label{e1013}
x^j_{n+1}=A_n x^j_n, \quad n\ge n_0.
\end{equation}
 We want to understand the behavior of $x^j_n, \ j=1,\dots,d$, as $n\to +\infty$. The answer is given by the following theorem.
\begin{theorem}[{Janas-Moszy\'nski \cite[Thm. 1.5]{jamo1}}]\label{t02}
In the above notation, let
$$
A_n=V_n+R_n,
$$
and $\{\l_j(n)\}_{j=1,\dots,d}$ be the eigenvalues of $V_n\in\cm_{d,d}(\bc)$. Suppose also that:
\begin{enumerate}
\item $\det A_n\not=0, \ \det V_n\not=0$ for $n\ge n_0$,
\item  $\{R_n\}\in \ell^1(\cm_{d,d}(\bc))$, and $\{V_n\}$ is of bounded variation, \emph{i.e.,} 
$$
||\{V_n\}||_{BV}:=\sum_n ||V_{n+1}-V_n||<\infty.
$$

\item the limit $V_\infty:=\lim_{n\to+\infty} V_n$ has non-zero distinct eigenvalues $\{\l_j(\infty)\}_{j=1,\dots,d},$ $|\l_j(\infty)|\not =|\l_k(\infty)|$, and
$$
\lim_{n\to+\infty} \l_j(n)=\l_j(\infty).
$$
The eigenvectors corresponding to $\{\l_j(\infty)\}_{j=1,\dots,d}$, are denoted by $\{v^j(\infty)\}$$_{j=1,\dots,d}$, $v_j(\infty)\in\bc^d$.
\end{enumerate}

Then there is a basis $x^j=\{x^j_n\}, \ j=1,\dots,d$, for the solutions of \eqref{e1013} such that
$$
x^j_n=\lp \prod_{k=n_0}^{n-1} \l_j(k)\rp\, \lp v^j(\infty) +\bar o(1)\rp,\quad n\to+\infty.
$$
\end{theorem}

Similar results for Jacobi matrices in the critical double root case are in Janas-Naboko-Shernova \cite{jana5}.

\subsection{A version of Schur-Frobenius lemma}\label{s14}
To give a classical version of Schur-Frobenius lemma, we introduce some notation. Let $H$ be a separable Hilbert space, $H_1$ be its closed subspace, and $H_2=H_1^\perp$ so that $H=H_1\oplus H_2$. Let $\cb(H)$ denote the algebra of bounded operators on $H$, and $\ca\in\cb(H), \ \ca^*=\ca$. Let 
\begin{equation}\label{e1014}
\ca=
\begin{bmatrix}
A&B\\
B^*&C 
\end{bmatrix}: H_1\oplus H_2\to H_1\oplus H_2
\end{equation}
be its block representation with respect to the above orthogonal decomposition of $H$. In particular, $A^*=A$ and $C^*=C$. The orthogonal projectors on $H_1, H_2$ are denoted by $P_{H_1}, P_{H_2}$, respectively.

The following useful proposition is well-known, see Tretter \cite[Sect. 1.6]{tre} for instance.
\begin{lemma}[{Schur-Frobenius lemma}]\label{l1} Let $\ca\in\cb(H)$ be as above and suppose that $C\in\cb(H_2)$ is boundedly invertible, \emph{i.e.} $C^{-1}\in \cb(H_2)$. 

Then the  operator $\ca$ is positive, $\ca\ge 0$, if and only if:
\begin{enumerate}
\item $A\ge 0_{H_1}$ and $C\ge 0_{H_2}$,
\item $A-B C^{-1} B^*\ge 0_{H_1}$.
\end{enumerate}
\end{lemma}

We need a generalization of the above result on Fredholm operators which is most probably a part of folklore on the subject. We give its proof for the completeness of the presentation. Let $\cf=\cf(H)$ be the ideal of finite rank operators on $H$, and $\bff=\bff(H)=\cb(H)/\cf$ be the factor-algebra of bounded operators on $H$ modulo $\cf$. For $X, Y\in\cb(H)$ we write $X\eeq Y$, if $X=Y$ as elements of $\bff$, or, equivalently, $X-Y\in\cf$. We say that $X$ equals to $Y$ modulo $\cf$. We say that $X\gee Y$ ($X\gse Y$), if 
$X-X^*, Y-Y^*\in\cf$, and $X-Y+F\ge 0$ ($X-Y+F>0$, respectively) for some $F\in\cf$. Similarly, $X\lee Y$ ($X\lse Y$), if there is a $F\in\cf(H)$ with the property $X-Y+F\le 0$ ($X-Y+F<0$, respectively). As usual, an operator $C\in\cb(H)$ is invertible in $\bff$, if there is a $C_1\in\cb(H)$ such that $C C_1\eeq I_H, C_1 C\eeq I_H$. More generally, $X$ satisfies a property modulo $\cf$, if $X+F$ has the claimed property for some $F\in\cf$.

\begin{lemma}[{Generalized Schur-Frobenius lemma}] \label{l2} Let $\ca\in \cb(H)$ be an operator as above. Suppose that $A=A^*$  in $\bff(H_1)$, $C=C^*$ in  $\bff(H_2)$, and $C$ is invertible in $\bff(H_2)$.

Then $\ca\gee 0$ if and only if:
\begin{enumerate}
\item $A\gee 0$ and $C\gee 0$,
\item $A-BC_1B^*\gee 0$, where $C_1$ is  the inverse of $C$ modulo $\cf$.
\end{enumerate}
\end{lemma}

\begin{proof}
We start with several auxiliary facts. First, we have $C\gee 0$, and for an $F\in \cf(H_2)$, we see that $C+F\ge 0$. So, replacing $C$ with $C+F$, we can suppose $C\ge 0$.

Second, under the assumptions of the lemma, $C$ and $C_1$ can be chosen self-adjoint without loss of generality. In fact, we have 
\begin{equation}\label{e1015}
C_1C\eeq CC_1\eeq I_{H_2},
\end{equation} 
and so $C^*C_1^*\eeq C_1^*C^*\eeq  I_{H_2}$. Since $C\eeq C^*$,  one can replace the operator $C$ by $\re C$ in the above relation. That is,
$$
C_1(\re C)\eeq (\re C)C_1\eeq I_{H_2},
$$
and we can suppose $C$ to be self-adjoint, $C=C^*$. Exactly the same argument applies to $C_1$.

Third, we have $C\ge 0_{H_2}$, but we can assume $C$ to be strictly positive without loss of generality. Indeed, relation \eqref{e1015} implies
$$
CC_1=I_{H_2}+F_1,\quad C_1C=I_{H_2}+F_2
$$
for certain $F_1,F_2\in \cf(H_2)$. Above,  $C, C_1$ are self-adjoint. In particular, $\kker C\subset \kker (I_{H_2}+F_2)$, and so $\dim\kker C$ is of finite dimension. It is easy to see that, for a fixed $\ep>0$, $C\eeq C+\ep P_{\kker C}$, and so we can replace $C$ with the latter operator. 

Let us prove that this new $C$ is strictly positive, \ie $C>0_{H_2}$. This is equivalent to show that $\s(C)\subset (0,+\infty)$, or $0\not\in \s(C)$. By contradiction, suppose that this is not the case, and $0\in \s(C)$. Since $\kker C=0$, this means that there is a sequence $\{f_n\}\subset H_2, ||f_n||=1$, such that $Cf_n\to 0$ and $f_n\stackrel w\to 0$ as $n\to +\infty$. As usual, the symbol $\stackrel w\to$ stays for the weak convergence. Consequently,
\begin{equation}\label{e101}
C_1Cf_n=(I_{H_2}+F_2)f_n=f_n+ F_2 f_n,
\end{equation}
as $n\to+\infty$. Notice  that the weak convergence of the sequence $\{f_n\}$ implies that $F_2 f_n\to 0$, since $F_2\in \cf(H_2)$. The LHS of \eqref{e101} goes to 0 by the choice of $\{f_n\}$, while the RHS of this relation is
$$
||f_n+F_2 f_n||\ge ||f_n||-||F_2f_n||\ge \frac12
$$
for $n$ large enough. Thus, we indeed have $C>0_{H_2}$, or, more precisely, $C>\ep I_{H_2}$.

We go now to the proof of the claim of the lemma. 

Let us first prove the direct implication. Recalling notation \eqref{e1014}, we have that $\ca\gee 0_H$, or there is a $F_0\in\cf(H)$ so that $\ca+F_0\ge 0_H$. For $f_1\in H_1$ we get
$\la \ca f,f\ra=\la (A+F_0)f,f\ra \ge 0$, consequently $A+P_{H_1}F_0 P_{H_1}\ge 0_{H_1}$, and $A\gee 0_{H_1}$. Similarly one obtains $C\gee 0_{H_2}$. Then, since $\ca\gee 0_H$, we also have
$$
\begin{bmatrix}
I& -BC_1\\
0& 0
\end{bmatrix}
\begin{bmatrix}
A& B\\
B^*& C
\end{bmatrix}
\begin{bmatrix}
I& 0\\
-C_1^*B^*& 0
\end{bmatrix}
\gee 0_H.
$$
Computing the operator on the LHS of this relation, we get
$$
\begin{bmatrix}
A-BC_1B^*&  0\\
0& 0
\end{bmatrix}\gee 0_H,
$$
which implies $A-BC_1B^*\gee 0_{H_1}$.

Turning to the proof of the inverse implication, we look at 
\begin{eqnarray*}
\la \ca f, f\ra&=&
\lt\la
\begin{bmatrix}
A& B\\
B^*& C
\end{bmatrix}
\,
\begin{bmatrix}
f_1\\
f_2
\end{bmatrix},
\,
\begin{bmatrix}
f_1\\
f_2
\end{bmatrix},
\rt\ra\\
&=&\la Af_1,f_1\ra+2\re\la f_2,B^*f_1\ra+\la C f_2,f_2\ra,
\end{eqnarray*}
where for an arbitrary $f\in H$ we set $f=f_1\oplus f_2,\ f_1\in H_1, f_2\in H_2$. By the hypotheses of the lemma,
\begin{eqnarray}
\la Af_1,f_1\ra&+&2\re\la f_2,B^*f_1\ra+\la C f_2,f_2\ra \nn\\
&\ge&\la BC_1B^* f_1,f_1\ra + 2\re \la f_2, B^*f_1\ra +\la C f_2, f_2\ra +\la F f,f\ra, \label{e10301}
\end{eqnarray}
where $F\in \cf(H)$. Assuming $C>0_{H_2}$,
let $g:=B^* f_1$ and then $h:=C^{-1/2}g=C^{-1/2}B^* f_1, \ e:=C^{1/2} f_2$. We continue as
\begin{eqnarray*}
\la BC_1B^* f_1,f_1\ra &+&2\re \la f_2, B^*f_1\ra +\la C f_2, f_2\ra=\la C_1 g,g\ra + 2\re \la e, C^{-1/2}g\ra +||e||^2 \nn\\
&=&\la C_1 C^{1/2} h, C^{1/2} h\ra+2\re\la e,h\ra+||e||^2. 
\end{eqnarray*}
Now, $C^{1/2} C_1 C^{1/2}\eeq I_{H_2}$, so quadratic form \eqref{e10301} equals to $||h+e||^2\ge 0$ modulo $\cf$.

Hence, we obtain that $\ca\gee 0$, and the proof of the lemma is complete.
\end{proof}

It is worth mentioning that the above reasoning gives a similar result for Calkin algebra, \ie the factor space of $\cb(H)$ over the compact operators instead the finite rank ones.

\section{Proof of Theorem \ref{t1}}\label{s01}
We are mainly interested in sufficient conditions on the presence of the discrete spectrum for block Jacobi matrices $J$ \eqref{e1}, as well as in an upper bound for the counting function of eigenvalues of $J$ in this situation.

Suppose that the spectrum $\s(J)$ of $J$ is purely discrete in the interval $I'$. We denote by $n(J; I'')$ the counting function of the eigenvalues of $J$ in the interval $I'', \ I''\subset I'$. Of course, the eigenvalues are numbered taking into account their multiplicity, \ie
\begin{equation}\label{e2}
n(J; I''):=\#(\s(J)\cap I'')=\#\{j: \l_j(J)\in I''\}.
\end{equation}
In general, we deal with unbounded operators having a dense domain containing $\ell^2_{d,\emptyset} \subset \ld$, see \eqref{e1016}; generically, we apply them to vectors from $\ell^2_{d,\emptyset}$, and so the expressions we manipulate are well-defined.

\begin{theorem}[{=Theorem A}]\label{t1} Let $J$ be a block Jacobi matrix \eqref{e1}. Suppose that:
\begin{enumerate} 
\item for some fixed $c\in\br$, there is a $N_0$ such that $B_{2n}\le c I_d$ for $n\ge N_0$.  
\item we have
$$
\lim_{n\to\infty} B_{2n-1}=+\infty
$$ 
in the sense of \eqref{e102-}.
\end{enumerate}
Then the spectrum of $J$ (denoting an arbitrary self-adjoint extension of $J_{min}$) is discrete in $(c,+\infty)$, and it accumulates to $+\infty$ only.

Moreover, for $a>0$
\begin{equation}\label{e3}
n(J; (c, c+a))\le d\cdot \cn(c+a)+\ind J_{min},
\end{equation}
where the quantity $\cn(c+a)$ is defined in \eqref{e6}, \eqref{e71}.
\end{theorem} 

\begin{proof}
Without loss of generality, we suppose that $c=0$, so the first condition of the theorem reads as $B_{2n}\le 0$ for $n\ge N_0$. Set $B=\mathrm{diag}\, \{B_j\}$ and $J_A=J-B$, that is
\begin{equation}\label{e4}
B=
\begin{bmatrix}
B_1& 0& 0&\ldots\\
0& B_2& 0&\ldots\\
0& 0& B_3&\ddots\\
\vdots& \ddots&\ddots&\ddots
\end{bmatrix}, \quad
J_A=
\begin{bmatrix}
0& A_1& 0&\ldots\\
A_1^*& 0& A_2&\ddots\\
0& 0& A_2^*&\ddots\\
\vdots& \ddots&\ddots&\ddots
\end{bmatrix}.
\end{equation}
For an $a>0$, the block decomposition of the operator $(J-aI)$ with respect to $\ell^2_d=\ldo\oplus\lde$ \ is
$$
(J-aI)=
\begin{bmatrix}
P_o (B-aI) P_o& P_o J_A P_e \\
P_e J_AP_o& P_e (B-aI) P_e \\
\end{bmatrix},
$$
where $P_o J_A P_o=P_e J_A P_e=0$. Therefore, we have 
\begin{eqnarray*}
(J-aI)^2-a^2I
&=&
\lt[
\begin{array}{c}
P_o (B-aI)^2 P_o-a^2 P_o+(P_o J_A P_e)(P_e J_A P_o)\\
(P_e J_A P_o)(P_o (B-aI) P_o)+ (P_e (B-aI) P_e) (P_e J_A P_o)
\end{array}\rt.\\
&&\lt.\begin{array}{c}
(P_o (B-aI)P_o) (P_o J_A P_e)+(P_oJ_A P_e)(P_e (B-aI) P_e)\\
(P_e J_A P_o)(P_o J_A P_e) + P_e (B-aI)^2 P_e - a^2 P_e
\end{array}\rt]
\end{eqnarray*}
To simplify the writing, set $T:=P_e J_A P_o$. The above equality takes the form
\begin{eqnarray}\label{e401}
(J-aI)^2-a^2I&=&
\lt[
\begin{array}{c}
P_o (B^2-2 aB) P_o+ T^*T\\
T (P_o (B-aI) P_o)+ (P_e (B-aI) P_e) T
\end{array}\rt. \\
&&
\lt.
\begin{array}{c}
P_o (B-aI)P_o T^*+ T^* P_e (B-aI) P_e\\
TT^* + P_e (B^2-2aB)^2 P_e
\end{array}\rt] \nn
\end{eqnarray}

Let 
\begin{eqnarray}\label{e5}
\cl_N&=&\{u=\{u_j\}\in \ld: u_j=0_d, \ j>N\}\subset \ld,\\
\cl_N^\perp&=&\{u=\{u_j\}\in \ld: u_j=0_d, \ j\le N\}\subset \ld.\nn
\end{eqnarray}
{\it The first step of the proof} is to show that the quadratic form $(J-aI)^2-a^2 I$ is positive on 
$\ell^2_{d, \emptyset}\cap \cl_N^\perp$ for $N$ large enough. Indeed, writing $u=u_o\oplus u_e, \ u\in \ld, u_o\in\ldo, u_e\in\lde$, we obtain
\begin{eqnarray*}\label{e501}
&&\la [(J-aI)^2-a^2I]u,u\ra=\lt\la  [(J-aI)^2-a^2I]
\begin{bmatrix}
u_o \\
u_e
\end{bmatrix},
\begin{bmatrix}
u_o \\
u_e
\end{bmatrix}\rt\ra \\
&=&\la P_o(B^2-2aB)P_o u_o,u_o\ra+ ||Tu_o||^2
+\la P_e(B^2-2aB)P_e u_e, u_e\ra+||T^*u_e||^2 
\\
&+&2\re \la(P_o(B-aI)P_o) T^*u_e,u_o\ra+2\re \la T^*(P_e(B-aI) P_e) u_e, u_o\ra.
\end{eqnarray*}
Recalling that $\la u_e, Tu_o\ra=\la T^* u_e, u_o\ra$, an elementary transformation yields
\begin{eqnarray}\label{e502}
\dots\ &=&\la P_o(B^2-2aB)P_o u_o,u_o\ra+||Tu_o||^2+||T^*u_e||^2\\
&+&2\rre \la T^*u_e, (P_o(B-2aI)P_o)u_o\ra +
2\rre \la (P_e B P_e) u_e, T u_o\ra \nonumber\\
&+&\la P_e(B^2-2aB)P_e u_e, u_e\ra. \nonumber
\end{eqnarray}
Suppose that $u=u_o\oplus u_e\in \cl_N^\perp$ with $N$ large enough. The choice of $N$ will be made precise below, see \eqref{e6}. By Cauchy-Schwarz inequality and using the fact that $B$ commutes with $P_o, P_e$, we get
\begin{eqnarray}
\dots&\ge&\la P_o(B^2-2aB)P_o u_o,u_o\ra+||Tu_o||^2+||T^*u_e||^2
+\la P_e(B^2-2aB)P_e u_e, u_e\ra \nonumber\\
&&-( ||P_o(B-2aI)P_o)u_o||^2+ ||T^*u_e||^2)-(||(P_e B P_e) u_e||^2+||T u_o||^2)
\nonumber\\
&=&\la P_o(B^2-2aB)P_o u_o,u_o\ra+\la P_e(B^2-2aB)P_e u_e, u_e\ra\nonumber\\
&&-\la  P_o(B-2aI)^2P_o u_o,u_o\ra -\la P_e B^2 P_e u_e, u_e\ra\nonumber\\
&=&\la P_o (2aB-4a^2)P_o u_o,u_o\ra+\la P_e (-2aB) P_e u_e,u_e\ra. \label{e51}
\end{eqnarray}
First, we require that $j\ge N_0$, so $B_{2j}\le 0$ and $\la P_e (-2aB) P_e u_e,u_e\ra\ge 0$ for  $u_e\in\cl_{2N_0}^\perp$. Second, we take $M=2a, a>0$ from definition \eqref{e102-} and, consequently, $B_{2j-1}\ge 2a I_d$ for $j\ge N(2a)$. That is, for 
$u_o\in \cl_{2N(2a)-1}^\perp$, we have $\la P_o (B-2a) P_o u_o,u_o\ra\ge 0$. So, 
$$
\la P_o (2aB-4a^2)P_o u_o,u_o\ra+\la P_e (-2aB) P_e u_e,u_e \ra \ge 2a \la P_o (B-2a) P_o u_o,u_o\ra\ge 0.
$$
To sum up, we define
\begin{equation}\label{e6}
\cn(2a):=\max\{2N_0, 2N(2a)-1\},
\end{equation}
and we see that
$$
\la [(J-aI)^2-a^2I]u,u\ra\ge 0
$$
for $u\in \ell^2_{d,\emptyset}\cap \cl_{\cn(2a)}^\perp$. By the min-max principle for self-adjoint operators \cite[Ch. 4]{birso1}, the total multiplicity of the spectrum $\s(J)$ in the interval $(0,2a)$ does not exceed $d\, \cn(2a)+\ind J_{min}$. Consequently, $\s(J)$ is discrete there, and
\begin{equation}\label{e7}
n(J; (0,2a))\le d\cdot \cn(2a)+\ind J_{min}.
\end{equation}

{\it The second step of the proof} consists in doing the similar computation for an arbitrary $c\in\br$. In fact, if the Jacobi matrix $J$ satisfies the assumptions of the theorem with a given $c$, consider $J\dpp:=J-c$, and apply the calculation from the first step of the proof to this operator.  The $d\times d$ matrices with double-primes refer to the matrix entries of $J\dpp$. Then  $B_{2j}\dpp=B_{2j}-c I_d\le0_d$, $B_{2j-1}\dpp=B_{2j-1}-c I_d$, and $A_j\dpp=A_j$.
Now, apply the result of the above first step to $J\dpp$, and then shift the spectrum $\s(J\dpp)$ by $+c$ to obtain
\begin{equation}\label{e71}
n(J; (c, c+2a))\le d\cdot \cn(c+2a)+\ind J_{min}.
\end{equation}
The theorem is proved.
\end{proof}

Of course, one may rewrite the conclusion of the theorem in a slightly different manner. For instance, we have that $\s(J)\cap (c,+\infty)=\s_d(J)\cap (c,+\infty)=\{\l_n\}$, where the eigenvalues $\l_n$ are  numbered increasingly counting the multiplicities. Then, picking $a>0$ in a way that $c+a=\l_n$ for some fixed $n$, we see that, for $a=\l_n-c$, relation \eqref{e71} becomes 
$$
n=n(J; (c,\l_n))\le d\cdot N(\l_n)+\ind J_{min},
$$
which implies a bound on $\{\l_n\}$ from below.

Here is another corollary of Theorem \ref{t1}.
\begin{proposition}\label{p1}
Let $J$ be a block Jacobi matrix \eqref{e1}.  Suppose that
\begin{equation}\label{e200}
\lim_{j\to\infty} B_{2j-1}=+\infty, \qquad \lim_{j\to\infty} B_{2j}= -\infty.
\end{equation}
Then the spectrum of $J$ is purely discrete and it accumulates to $\pm\infty$ only.
\end{proposition}

Of course, relation \eqref{e200} means that for any $M>0$, there is a $N=N(M)$ such that
$B_{2j-1}\ge M I_d$ and $B_{2j}\le -M I_d$ for $j\ge N$.

\begin{proof}
By the hypotheses of the proposition, we have that $B_{2j-1}\to +\infty$ as $j\to+\infty$ and $B_{2j}\le -M I_d$ for some fixed $M>0$ and $j\ge N_0(M)$. So, by Theorem \ref{t1}, the spectrum $\s(J)$ is discrete in $(-M,+\infty)$. Since $M>0$ is arbitrary, the whole $\s(J)$ is discrete.

The spectrum has to accumulate to $\pm \infty$, since
$$
\la B_{2j-1} u_{2j-1}, u_{2j-1}\ra=\la J u, u\ra \to +\infty
$$
for $u_{2j-1}\in (\bc^d)_{2j-1}, ||u_{2j-1}||=1$ and
$$
u=\{u_k\}=\lt\{
\begin{array}{lcl}
0_d&,& k\not=j-1,\\
u_{2j-1}&,& k=2j-1.
\end{array}
\rt.
$$ 
A similar bound for $u_{2j}\in (\bc^d)_{2j}, ||u_{2j}||=1$
$$
\la B_{2j} u_{2j}, u_{2j}\ra=\la J u, u \ra \to -\infty, \quad j\to+\infty,
$$
finishes the proof. Above, 
$$
u=\{u_k\}=\lt\{
\begin{array}{lcl}
0_d&,& k\not=j-1,\\
u_{2j}&,& k=2j.
\end{array}
\rt.
$$ 
\end{proof}

\section{Proof of Theorem \ref{t22}}\label{s2}
First we prove a version of Theorem \ref{t22} for Jacobi matrices with invertible diagonal. Then we show that one can remove this assumption, and, modulo suitable modification,  the result still holds true.

For $B_n\in\cm_{d,d}(\bc), B_n=B_n^*$,  the (elementary) spectral theorem allows us to define $B_{n+}=(B_n)_+:=B_n\, P_{(0,+\infty)}, \ B_{n-}=(B_n)_-:=B_n\, P_{(-\infty, 0]}$, where $P_I$ is the spectral projection of $B_n$ on the interval $I\subset \br$, so that $B_n=B_{n+}+B_{n-}$. 

\begin{theorem}\label{t21} Let $J$ be the block Jacobi matrix defined by \eqref{e1}. Suppose that:
\begin{enumerate}
\item $B$ is invertible, \ie $0\not\in \s(B)$,
\item $B_{2k-1}\to +\infty$ as $k\to +\infty$,
\item $(B_{2k})_+\big |_{\im((B_{2k})_+)}\to+\infty$ as $k\to +\infty$, or $(B_{2k})_+=0$,
\item one has 
\begin{equation}\label{e80}
\limsup_{n\to+\infty}||(B^{-1/2}_{2n})_+ A^*_{2n-1}B^{-1/2}_{2n-1}||+
\limsup_{n\to+\infty}||(B^{-1/2}_{2n-2})_+ A_{2n-2} B^{-1/2}_{2n-1}||<1.
\end{equation}
\end{enumerate}
Then the part of the spectrum $\s(J)\cap(0, +\infty)$ is discrete.
\end{theorem}

Of course, one can interpret assumption (3) of the theorem exactly as $\lim_{k\to+\infty}$ $\opp{B^{-1}_{2k}}=0$. Similarly to the discussion at the end the proof of Theorem \ref{t1}, one can give an upper bound on the eigenvalue counting function $n(J; I)$ for an interval $I$, see \eqref{e2}.

\begin{proof} As in Theorem \ref{t1}, the idea is to prove the inequality $(J-aI)^2-a^2I\gee 0$ for a parameter $a>0$ going to $+\infty$. This implies that $\s(J)\cap (0,2a)$ is discrete, which will  give immediately the claim of the theorem.

Recall that $T:=P_e J_A P_o$. We continue on with relation \eqref{e502} from Theorem \ref{t1}
\begin{eqnarray}\label{e9}
&&\\
\la [(J-aI)^2-a^2I]u,u\ra&=&
\la P_o(B^2-2aB)P_o u_o,u_o\ra+ ||Tu_o||^2 \nn\\
&+&\la P_e(B^2-2aB)P_e u_e, u_e\ra+||T^*u_e||^2 \nn \\
&+&2\re \la T^*u_e, (P_o(B-2aI)P_o) u_o\ra \nn\\
&+&2\re \la T u_o, (P_e B P_e) u_e \ra, \nn
\nonumber
\end{eqnarray}
where $u=u_o\oplus u_e\in\ld$. By Cauchy-Schwarz inequality, we get
$$
2\re \la T^*u_e, (P_o(B-2aI)P_o) u_o\ra
\ge -\big[ ||T^*u_e||^2+|| (P_o(B-2aI)P_o) u_o||^2\big].
$$
So, we can bound \eqref{e9} from below as
\begin{eqnarray*}
\la [(J-aI)^2-a^2I]u,u\ra&\ge &2a\la P_o(B-2a) P_o u_o, u_o\ra+||Tu_o||^2\\
&+&2\re \la T^*(P_e B P_e) u_e, u_o \ra + \la P_e (B^2-2aB) P_e u_e, u_e\ra.
\end{eqnarray*}
We see that the latter quadratic form is generated by the following operator written in the $2\times 2$-block decomposition
\begin{equation}\label{e1002}
\begin{bmatrix}
\ti A &\ti B \\
\ti B^*& \ti C
\end{bmatrix}
:=
\begin{bmatrix}
2a(P_o (B-2a) P_o)+T^*T & T^*(P_e B P_e) \\
(P_e B P_e) T& P_e(B^2-2aB) P_e 
\end{bmatrix}
: \ldo\oplus\lde\to\ldo\oplus\lde.
\end{equation}
To see the positivity modulo $\cf$ of the above operator, we apply Lemma \ref{l2}. Trivially, we have
$$
\ti A=2a(P_o (B-2a) P_o)+T^*T\gee 0
$$
provided assumption (2) of the theorem. In the same way, we see that
$$
\ti C= P_e(B^2-2aB) P_e\ge P_e(B^2_+-2aB_+) P_e\gse 0,
$$
since, by assumption (3) of the theorem $(B_{2k})_+\to +\infty$ (in the sense that $\min (\s(B_{2k}\cap (0, +\infty))\to +\infty$ as $k\to+\infty$). 

Assumption (1) yields that $B$ is invertible, and so is $B_+$ on $\im(B_+)$. Furthermore, by assumptions (2), (3) the set $\s(B)\cap (0,+\infty)$ is discrete and it accumulates to $+\infty$ only. Without loss of generality, we can take $2a\not\in \s(B)\cap (0,+\infty), a>0$. The spectral mapping theorem says then that $(B^2-2aB)$ is invertible, and so is $P_e(B^2-2aB)P_e$ on $\lde$. The  same applies to $P_e(B^2_+-2aB_+)P_e$.

Now, we have to check that the operator
\begin{eqnarray}\label{e10}
\ti A-\ti B \ti C^{-1}\ti B^*&=&2a (P_o(B-2a)P_o)+T^*T \nn\\
&-&T^*(P_eBP_e)(P_e(B^2-2aB)P_e)^{-1}(P_eBP_e)T\nn\\
&=&2a (P_o(B-2a)P_o +T^*\big[I_e-(P_eB^2P_e)(P_e(B^2-2aB)P_e)^{-1}P_e\big]T
\end{eqnarray}
is positive modulo $\cf$. Using the fact that $B$ is a block-diagonal operator, we have $P_eB=BP_e$, and, after a simple algebraic calculation, \eqref{e10} takes the form
$$
\dots\ =2a (P_o(B-2a)P_o)-2a T^*(P_e(B-2a)^{-1}P_e)T.
$$
Its positvity is equivalent to the positivity of
$$
(P_o(B-2a)P_o)-T^*(P_e(B-2a)^{-1}P_e)T.
$$
Furthermore, we have $(B-2a)=(B-2a)_+\oplus (B-2a)_-$ and $(B-2a)^{-1}=(B-2a)_+^{-1}\oplus (B-2a)_-^{-1}\le (B-2a)_+^{-1}$. So, the required positivity will be proved, if we obtain
\begin{equation}\label{e1001}
P_o(B-2a)P_o-T^*(P_e(B-2a)_+^{-1}P_e) T\gee 0.
\end{equation}
Once again, the invertibility of $B$ yields that $(P_oBP_o)$ is invertible as well on $\ldo$.  Moreover, assumption (2) of the theorem implies that $(P_oBP_o)\gse0$. Let us conjugate relation \eqref{e1001} by $(P_0 BP_0)^{-1/2}$. Recalling $P_oB=BP_o$, this gives
\begin{equation}\label{e11}
P_oB^{-1}(B-2a)P_o-(P_oBP_o)^{-1/2}T^*(P_e(B-2a)_+^{-1}P_e)\, T(P_oBP_o)^{-1/2}\gee 0.
\end{equation}
Consider the operator $(B-2a)_+$ on $\lde$. One has 
$(B-2a)_+=(B_+-2a)_+= B_+(I-2a{B_+}^{-1})_+$, where $B_+^{-1}:=(B^{-1})_+$. Restricting the latter expression on $\cl^\perp_N$ with $N$ large enough, we get $||{B_+}^{-1}||<1/(2a)$, and so $(I-2aB_+^{-1})\ge 0$ on this subspace. Hence 
$$
P_e(B-2a)^{-1}_+P_e=P_e\big (B_+^{-1}(I-2aB_+^{-1})\big )_+^{-1}P_e=P_e B_+^{-1}(I-2aB_+^{-1})^{-1}P_e, 
$$ 
and we rewrite \eqref{e11} as
$$
P_o-2a(P_oB^{-1}P_o) -(P_oBP_o)^{-1/2}T^*(P_e B_+^{-1}(I-2aB^{-1}_+)^{-1}P_e) T(P_oBP_o)^{-1/2}\gee 0.
$$
Note that under conditions (2), (3) of the theorem, for any $\d>0$, we have
$$
B_+^{-1}(I-2aB_+^{-1})^{-1}\lee (1+\d)B_+^{-1}.
$$
Therefore, it is enough to check
\begin{equation}\label{e12}
P_o -(1+\d)\, (P_oBP_o)^{-1/2}T^*(P_e B_+^{-1}P_e)T(P_oBP_o)^{-1/2} -2a(P_oB^{-1}P_o)\gee 0.
\end{equation}
Take $N_1, \ N_1\ge N$ large enough to guarantee that the norm of $2a(P_oB^{-1}P_o)$ is as small as we want and restrict the latter inequality to $\cl^\perp_{N_1}$.  Picking $\d>0$ small enough, we see that inequality \eqref{e12} holds, if
$$
P_o -(P_oBP_o)^{-1/2}T^*(P_e B_+^{-1}P_e)T(P_oBP_o)^{-1/2}\gee \d_1P_o
$$
for some $\d_1>0$. In turn, this relation is true if we require
$$
|| (P_e B_+^{-1/2} P_e)T(P_o B P_o)^{-1/2}||<1,
$$
where we understand that the operator $(P_e B_+^{-1/2} P_e)T(P_o B P_o)^{-1/2}$ is restricted to $\cl^\perp_{N_2}$ for some $N_2\ge N_1$. This is the same as
\begin{equation}\label{e1200}
\limsup_{N_2\to+\infty} || (P_e B_+^{-1/2} P_e)T(P_o B P_o)^{-1/2}\Big|_{\cl^\perp_{N_2}}||<1.
\end{equation}
Properly understanding the inverse of the operator, we have
$$
(P_e B_+^{-1/2} P_e)=(P_e B_+^{-1} P_e)^{1/2}=(P_e B_+ P_e)^{-1/2}.
$$
Denote by $S$ the shift operator acting on $\ell^2_d$, \ie  $S\{u_n\}=\{0, u_1, u_2, \dots\}$ for $u\in\ell^2_d$. Set also $A=\mathrm{diag}\, \{A_j\}$, and so
$$
T=P_eJ_AP_o=P_eSA^*P_o+P_eAS^*P_o,
$$
see the formula preceding \eqref{e401}. Consequently, relation \eqref{e1200} follows from
\begin{eqnarray*}
&&\limsup_{n\to+\infty}||(P_e B_+^{-1/2} P_e) SA^* (P_o BP_o)^{-1/2}\Big|_{\cl^\perp_n}||\\
&&+\limsup_{n\to+\infty}||(P_e B_+^{-1/2} P_e) AS^* (P_o BP_o)^{-1/2}\Big|_{\cl^\perp_n}||<1.
\end{eqnarray*}
Explicitely computing the operators appearing under the norms in the previous relation gives
$$
\limsup_{n\to+\infty} ||(B^{-1/2}_{2n})_+ A^*_{2n-1} B_{2n-1}^{-1/2}||+
\limsup_{n\to+\infty} ||(B^{-1/2}_{2n-2})_+ A_{2n-2} B_{2n-1}^{-1/2}||< 1,
$$
which is the claim of the theorem.
\end{proof}

Now, the point is to obtain an extension of Theorem \ref{t21} getting rid of the technical condition (1) from its formulation. Recall the notation
$$
\opp{B^{-1}}=
\lt\{
\begin{array}{lcl}
(B_+)^{-1}&, & \ \mathrm{on} \ \im (B_+),\\
0_d&,&  \ \mathrm{on} \ \im (B_+)^\perp.
\end{array}
\rt.
$$
introduced in \eqref{e1031}. Recall that, with a slight abuse of the notation,  $\opp{B^{-1/2}}=\opp{B^{-1}} \hspace{-0,7mm}^{1/2}$.

\begin{theorem}[{=Theorem B}]\label{t22} Let $J$ be the Jacobi matrix defined by \eqref{e1}. Assume that:
\begin{enumerate}
\item $\lim_{k\to\infty}B_{2k-1}=+\infty$, 
\item either $(B_{2k})_+\to+\infty$ or $(B_{2k})_+=0$,
\item one has 
\begin{equation}\label{e8}
\limsup_{n\to+\infty} ||\opp{B^{-1/2}_{2n}}A^*_{2n-1} B^{-1/2}_{2n-1}||+
\limsup_{n\to+\infty} ||\opp{B^{-1/2}_{2n-2}}A_{2n-2} B^{-1/2}_{2n-1}||<1.
\end{equation}
\end{enumerate}
Then the part of the spectrum $\s(J)\cap(0, +\infty)$ is discrete.
\end{theorem}

Compared to Theorem \ref{t21}, we drop the condition of the invertibility of $B$. The price to pay is that the expressions $(B_{2k}^{-1})_+$ in \eqref{e80} become $\opp{B_{2k}^{-1}}$ in \eqref{e8}.
Notice that condition (2) of the theorem can be read as $\lim_{k\to+\infty} \opp{B^{-1}_{2k}}=0$.
.

\begin{proof} The proof of the theorem follows exactly the lines of Theorem \ref{t21}. The only thing to do is to explain  how one gets rid of the invertibility of $B$ in Theorem \ref{t21}.

The problem is that to apply the Schur-Frobenius lemma to \eqref{e1002}, we need to have an invertible $(P_e B_+P_e)$. To this end, recall that $\ld=\ldo\oplus\lde$,
and consider $E_{2k}=\im B_{2k}\subset (\bc^d)_{2k}$. Of course, since $B_{2k}$ is self-adjoint, $\im B_{2k}=(\bc^d)_{2k}\ominus \kker B_{2k}$. Set 
$$
\ce=\oplus_{k: E_{2k}\not =\{0\}} E_{2k}\subset \lde,
$$
and $\cg=\lde\ominus\ce$; the corresponding orthoprojectors are denoted by $P_\ce$ and $P_\cg$, respectively. It follows from the construction that $P_e B_+ P_e: \ce\to \ce$ is an invertible operator.

Now, let us turn back to the block representation of the quadratic form \eqref{e1002}, \ie 
\begin{equation}\label{e19}
\begin{bmatrix}
2a(P_o (B-2a) P_o)+T^*T & T^*(P_e B P_e) \\
(P_e B P_e) T& P_e(B^2-2aB) P_e 
\end{bmatrix}.
\end{equation}
Consider its restriction on the space $\ell^2_o\oplus\ce$; it is given by
\begin{eqnarray}\label{e20}
&&\begin{bmatrix}
2a(P_o (B-2a) P_o)+T^*T & T^*(P_e B P_e) P_\ce\\
P_\ce (P_e B P_e) T& P_\ce(B^2-2aB) P_\ce 
\end{bmatrix}\nonumber\\
&\ge&
\begin{bmatrix}
2a(P_o (B-2a) P_o)+T^*P_\ce T & T^*P_\ce (P_e B P_e) P_\ce\\
P_\ce (P_e B P_e) P_\ce T& P_\ce(B^2-2aB) P_\ce 
\end{bmatrix},
\end{eqnarray}
and the positivity of the latter form easily yields the positivity of quadratic form \eqref{e19}, since $T^*T=T^*(P_\ce+P_\cg) T\ge T^*P_\ce T$. We observe that the proof of Theorem \ref{t21} goes through for the form in the RHS of \eqref{e20} (\ie, with $T$ replaced with $P_\ce T$).  Hence, we obtain the discreteness of $\s(J)$ on the interval $(0,2a)$, and the proof of the theorem is finished.
\end{proof}
Condition (2) in Theorem \ref{t22} says that the sequence $\{B_{2k}\}$ is allowed to contain two subsequences with rather different behavior, \emph{i.e.} one can have $(B_{2k})_+\to+\infty$ for the first subsequence and $(B_{2k})_+=0$ for the second subsequence complementary to the first one. 

It goes without saying that, as in Theorem \ref{t1},  one can easily write down the ``shifted'' versions of Theorems \ref{t21}, \ref{t22}.

\section{Some examples to Theorems \ref{t1}, \ref{t21}, \ref{t22}}\label{s3}
In this section we give a series of examples discussing the phenomena described in Theorems \ref{t1}-\ref{t22}. Besides that, the examples show that the results obtained in these theorems are sharp, and they give negative answers to some natural attempts of their generalization. Most examples hold even in the scalar-valued case $d=1$.

\subsection{Example 0}\label{s301}
The following easy fact was mentioned in the introduction, and now we give a sketch of its proof for the completeness. Let
$$
J_\a=J(\{n^\a\},\{0\}),\quad J_{\a,\b}=J(\{n^\a\},\{n^\b\}),
$$
where $0<\a\le 1,\ 0<\b<\a$, see \eqref{e101-} for notation. Set also $B=\mathrm{diag}\ \{n^\b\}$. Then $J_{\a,\b}=J_\a +B$ is neither relatively compact nor relatively bounded perturbation of $J_\a$.

Observe that, due to Carleman condition \cite[Ch. 1, Pb. 1]{akh}, the Jacobi matrix $J_\a$ is in the limit point case, and so $\cdd(J_\a)=\cdd((J_\a)_{max})$, where $J_\a$ stays for the self-adjoint extension of $J_\a$ for the simplicity of notation.

We have to present an $u=\{u_n\}\in \ell^2$ such that $J_\a u\in \ell^2$, but $Bu\not\in\ell^2$. Set $u_n:=(i^n)/n^x, \ u:=\{u_n\}$. The condition $u\in\ell^2$ implies $x>1/2$. Then
\begin{eqnarray*}
(J_\a u)_n&=&(n-1)^\a\frac{i^{n-1}}{(n-1)^x}+n^\a\frac{i^{n+1}}{(n+1)^x}\\
&=&i^{n-1}n^{\a-x}\lb\frac{2x-\a}{n}+O\lp\frac 1{n^2}\rp\rb.
\end{eqnarray*}
Assuming $2x\not=\a$, the vector $J_\a u$ is in $\ell^2$ whenever $x>\a+1/2$, and this is true since $x>1/2$. Now, $Bu=\lt\{\frac{i^n}{n^{x-\b}}\rt\}$ is not in $\ell^2$ if $x\le \b+1/2$. So, for an arbitrary small  $\b>0$, one can choose a value $x$ so that $1/2<x\le 1/2+\b$, and the vector $u$ satisfies all requirements. Therefore, $B$ is not a relatively compact perturbation of $J_\a$, and it does change dramatically its essential spectrum. 

\subsection{Example 1}\label{s31} Consider the scalar-valued Jacobi matrix
$J=J(\{a_n\},\{b_n\})$ where
$$
a_n=n^\a, \quad
b_n=
\lt\{
\begin{array}{lcl}
bn^\a&,& n=2k-1,\\
0&,& n=2k,
\end{array}
\rt.
$$
where $2/3\le\a<1$ and $b>0$, and so $c=0$, see Theorem \ref{t1}. On the one hand, Damanik-Naboko \cite[Thm. 3]{dana} say that $\s(J)\cap (-\infty, 0]$ is purely absolutely continuous. On the other hand, Theorem \ref{t1} claims that $\s(J)\cap (0,+\infty)$ is purely discrete with only accumulation point at $+\infty$. Notice that the example shows that one cannot extend the presence of the discrete spectrum of $J$ to the interval $(-\ep,+\infty)$ for any $\ep>0$, and hence the value $c$ in Theorem \ref{t1} is sharp.

The above example also shows that Theorem \ref{t22} is sharp as well. Indeed, we have in terms of Theorem \ref{t22} that $B_{2n-1}=b_{2n-1}=b(2n-1)^\a\to+\infty$, and $B_{2n}=b_{2n}=0$. Hence $(B_{2n})_+=0$ and $\opp{B^{-1}_{2n}}=0$. As above, we have $\s(J)\cap(-\infty, 0]$ is purely absolutely continuous, and, by Theorem \ref{t22}, $\s(J)\cap (0,+\infty)$ is purely absolutely discrete.

\subsection{An heuristic example}\label{s32} A very special feature of Theorem \ref{t1} is that the result depends on the behavior of entries $B_{2n-1}, B_{2n}$ of the matrix $J$, but not on $A_n$ which can be arbitrary. A simple example presented below explains how this phenomenon can occur; it is of heuristic nature and it cannot be considered as a rigorous proof. Since the following discussion makes quite apparent that the effect comes from a fine cancellation  of ``$A_n$-terms'' in certain asymptotics, we decided to include it in the paper.

Consider  the scalar-valued Jacobi matrix $J=J(\{a_n\},\{b_n\})$ with entries
$$
a_n=n^\a, \quad 
b_n=
\lt\{
\begin{array}{lcl}
n^\b&,& n=2k-1,\\
0&,& n=2k,
\end{array}
\right.
$$
where $1/2<\a<1, \ 0<\b<\a$, and $\a>\frac 12\max\{1+\b, 2-\b\}$. The last relation implies in fact that $3/4\le\a<1$. Denote by $\pp_n$ the standard transfer matrix for $J$, \ie
$$
\begin{bmatrix}
u_n\\ u_{n+1}
\end{bmatrix}
=\pp_n
\begin{bmatrix}
u_{n-1}\\ u_n
\end{bmatrix},
$$
where $(J u^\pm)’(\l)=\l (u^\pm)’(\l)$, and $u^\pm:=u^\pm(\l)=\{u^\pm_n\}$, see Subsection \ref{s12} for the prime-notation.

We have for 2-step transfer matrices
\begin{eqnarray}\label{e14}
\pps_n&:=&\pp_{2n}\pp_{2n-1} \nn\\
&=&
\begin{bmatrix}
0& 1\\
-\frac{(2n-1)^\a}{(2n)^\a}&\frac\l{(2n)^\a} 
\end{bmatrix}\,
\begin{bmatrix}
0& 1\\
-\frac{(2n-2)^\a}{(2n-1)^\a}&\frac{\l-b_{2n-1}}{(2n-1)^\a} 
\end{bmatrix}\nn \\
&=&
\begin{bmatrix}
-\frac{(2n-2)^\a}{(2n-1)^\a}& \frac{\l-b_{2n-1}}{(2n-1)^\a}\\
-\frac{\l(2n-2)^\a}{(2n)^\a(2n-1)^\a}&-\frac{(2n-1)^\a}{(2n)^\a}+\frac{\l(\l-b_{2n-1})}{(2n)^\a(2n-1)^\a}
\end{bmatrix}\nn\\
&=&-I_2+
\begin{bmatrix}
\frac\a{2n}&\frac\l{(2n)^\a}-\frac{b_{2n-1}}{(2n-1)^\a} \\
-\frac\l{(2n)^\a}& \frac\a{2n}-\frac{\l b_{2n-1}}{(2n)^{2\a}}
\end{bmatrix}
+R_n,\nn
\end{eqnarray}
where $\{R_n\}\in \ell^1(\cm_{d,d}(\bc))$ and we used the Taylor expansion of elementary functions appearing in the preceding expression. To calculate the eigenvalues $\l_\pm(n)$ of $\pps_n$, we notice that
\begin{eqnarray*}
\det \pps_n&=&\det\pp_{2n}\det\pp_{2n-1}=\lp\frac{2n-2}{2n}\rp^\a=\lp 1-\frac1n\rp^\a,\\
\tr \pps_n&=&-2+\frac\a n-\frac{\l b_{2n-1}}{(2n)^{2\a}}+O\lp\frac1{n^{2\a}}\rp,
\end{eqnarray*}
and, consequently,
\begin{eqnarray}\label{e1004}
\l_\pm(n)&=&\frac12\tr\pps_n\pm\sqrt{\frac14(\tr\pps_n)^2-\det\pps_n} \nn\\
&=&-1+\frac\a{2n}-\frac{\l b_{2n-1}}{2(2n)^{2\a}}+O\lp\frac1{n^{2\a}}\rp \nn\\
&\pm&\sqrt{\lp-1+\frac\a{2n}-\frac{\l b_{2n-1}}{2(2n)^{2\a}}+O\lp\frac1{n^{2\a}}\rp\rp^2-\lp 1-\frac1n\rp^\a}\\
&=&-1+\frac\a{2n}-\frac{\l b_{2n-1}}{2(2n)^{2\a}}+O\lp\frac1{n^{2\a}}\rp \nn\\
&\pm&\frac{\sqrt{\l b_{2n-1}}}{(2n)^\a}+O\lp\frac 1{\sqrt{b_{2n-1}}n^\a}\rp. \nn
\end{eqnarray}
Hence we obtain
$$
\l_\pm(n)=-1+\frac\a{2n}\pm\frac{\sqrt{\l b_{2n-1}}}{(2n)^\a}+r_n,
$$
where
$$
r_n=-\frac{\l b_{2n-1}}{2(2n)^{2\a}}+O\lp\frac1{n^{\a+\b/2}}\rp.
$$
Take $\l>0$. The condition $2\a>\max\{1+\b, 2-\b\}$ yields that $\{r_n\}\in \ell^1$. Recall that $u^\pm:=u^\pm(\l)=\{u^\pm_n\}$ is a generalized eigenvector for the equation $(Ju^\pm)’(\l)=\l (u^\pm)’(\l)$. A Levinson type result in this situation, \emph{e.g.}, see Theorem \ref{t02}, along with a simple computation would imply formally that
\begin{eqnarray*}
|u^-_{2N}|&=& (1+\bar o(1))\, \prod^N_{n=1} |\l_-(n)|
=(1+\bar o(1))\, \prod^N_{n=1} \lt| -1+\frac\a{2n}-\frac{\l^{1/2}}{(2n)^{\a-\b/2}}\rt|\\
&=& O\lp\frac1{n^{\a/2}}\exp\lb-\frac{\sqrt{\l}}{1-\a+\b/2}\, \lp\frac n2\rp^{1-\a+\b/2}\rb\rp.
\end{eqnarray*}
A similar asymptotic bound holds for odd-numbered components $u^-_{2N-1}$, and hence we get $u^-\in \ell^2$, which suggests  that $\l\in \s_d(J)$ whenever $u(\l)$ satisfies the appropriate boundary conditions, see Theorem \ref{t01}. On the contrary, taking $\l<0$ gives rise to bounded oscillating solutions $u^\pm=\{u^\pm_n\},\ (Ju^\pm)’(\l)=\l (u^\pm)’(\l)$, where, granting a Levinson-type result (Theorem \ref{t02}) again,
\begin{eqnarray*}
u^\pm_{2N}&=&(1+\bar o(1))\, \prod^N_{n=1} \l_\pm(n)
= (1+\bar o(1))\, \prod^N_{n=1} \lp-1+\frac\a{2n}\pm\frac{i|\l|^{1/2}}{(2n)^{\a-\b/2}}\rp\\
&=& O\lp\frac1{n^{\a/2}}\exp\lb \pm\frac{i\sqrt{|\l|}}{1-\a+\b/2}\, \lp\frac n2\rp^{1-\a+\b/2}\rb\rp.
\end{eqnarray*}
By subordinacy theory for eigenvectors of Jacobi matrices (Theorem \ref{t01}),  this suggests that $\l\in \s_{ac}(J)$.

The issue here is that a Levinson-type theorem is not fully applicable in this situation. Still, the point is that one clearly sees the cancellation of leading terms depending on $a_n$ (\ie $\a$-power terms) under the square root \eqref{e1004}, thus bringing into “the main game” a term depending on $b_{2n-1}$ producing either rapidly decreasing solutions $u_-(\l)$ for $\l>0$, or bounded oscillating solutions $u^\pm(\l)$ for $\l<0$. Theorem \ref{t1} gives a rigorous proof for this phenomenon in a completely general case.

Concluding this example, we remark that it is likely that a more rigorous analysis in this direction might be performed with the help of methods introduced in Naboko-Simonov \cite{nasi}.

\subsection{A counterexample for a counterpart of Theorem \ref{t1} “with step 3”}\label{s33}
Let $J=J(\{A_n\},\{B_n\})$ be a block Jacobi matrix. Theorem \ref{t1} says that, if the entries $B_{2n}$ are uniformly bounded and $B_{2n-1}\to +\infty$ as $n\to+\infty$, then the spectrum $\s(J)$ is discrete on a right half-axis. This claim prompts to a guess that the same effect is likely to take place if we extend the “2-step” formulation to a larger number of “steps”. For instance, we wish to understand if the conditions
$$
\sup_{n\not =n_k} B_n\le c_0,\quad B_{n_k}\to+\infty, \ k\to+\infty
$$
for a “reasonable” sequence $\{n_k\}\subset \bn, \ n_k\nearrow +\infty$, ensure that $\s(J)$ is discrete on a right real half-axis. 

It turns out that the answer to this guess is negative. Moreover, even the “3-step” assumptions do not give the discreteness conclusion in full generality, as the below proposition shows. Hence, the open problem is to give a version of “n-step” result similar to Theorem \ref{t1} under some additional assumptions.

\begin{proposition}\label{p3} There is a scalar-valued Jacobi matrix $J=J(\{a_n\},\{b_n\})$ such that
$$
b_{3n}\to+\infty, \quad b_{3n-1}=b_{3n-2}=0,
$$
and, for a suitable choice of $\{a_n\}$, the spectrum $\s(J)$ is purely absolutely continuous and $\s(J)=\s(J_{ac})=\br$.
\end{proposition}

\begin{proof} We choose
$$
a_n=n^\a, \quad b_{3n}=\d(3n)^\a,
$$
with $1/2<\a<1$ and $0<\d<2$. Notice that for $\d=2$, this choice corresponds to the critical coupling of the main diagonal $\{b_n\}$ and auxiliary diagonals $\{a_n\}$, that is
$$
a_{n-1}+a_n-b_n=\bar o(n^\a),
$$
see \cite{jana1, jana2}. We have for 3-step transfer matrices
\begin{eqnarray*}
\pps_n&:=&\pp_{3n}\pp_{3n-1}\pp_{3n-2}\\
&=&
\begin{bmatrix}
0& 1\\
-\frac{(3n-1)^\a}{(3n)^\a}& \frac{\l-b_{3n}}{(3n)^\a}
\end{bmatrix}
\,
\begin{bmatrix}
0& 1\\
-\frac{(3n-2)^\a}{(3n-1)^\a}& \frac\l{(3n-1)^\a}\\
\end{bmatrix}
\, 
\begin{bmatrix}
0& 1\\
-\frac{(3n-3)^\a}{(3n-2)^\a}& \frac\l{(3n-2)^\a}\\
\end{bmatrix}\\
&=&
\lt[
\begin{array}{l}
-\frac{\l(3n-3)^\a}{(3n-1)^\a(3n-2)^\a}\\
\frac{(3n-1)^\a(3n-3)^\a}{(3n)^\a(3n-2)^\a}-\frac{\l(\l-b_{3n})(3n-3)^\a}{(3n)^\a(3n-1)^\a(3n-2)^\a}
\end{array}
\rt.\\
&&
\lt.
\begin{array}{l}
-\frac{(3n-2)^\a}{(3n-1)^\a}+\frac{\l^2}{(3n-2)^\a(3n-1)^\a}\\
-\frac{\l(3n-1)^\a}{(3n)^\a(3n-2)^\a}+\frac{(\l-b_{3n})}{(3n)^\a}
\lb -\frac{(3n-2)^\a}{(3n-1)^\a}+\frac{\l^2}{(3n-1)^\a(3n-2)^\a}\rb
\end{array}
\rt],
\end{eqnarray*}
and
\begin{eqnarray*}
\det \pps_n&=&\det\pp_{3n}\, \det\pp_{3n-1}\, \det\pp_{3n-2}\\
&=&\frac{(3n-3)^\a}{(3n)^\a}=\lp 1-\frac1n \rp^\a,\\
\tr \pps_n&=&-\frac{\l(3n-3)^\a}{(3n-1)^\a(3n-2)^\a}-\frac{\l(3n-1)^\a}{(3n)^\a(3n-2)^\a}\\
&+&\frac{(\l-b_{3n})}{(3n)^\a}
\lb -\frac{(3n-2)^\a}{(3n-1)^\a}+\frac{\l^2}{(3n-1)^\a(3n-2)^\a}\rb\\
&=&\d-\frac{3\l}{(3n)^\a}-\frac{\a\d}{3n}+O\lp\frac 1{n^{2\a}}\rp.
\end{eqnarray*}
where we used once again Taylor expansions for standard power functions. Furthermore, expanding similarly the components of $\pps_n$, we obtain
\begin{eqnarray}\label{e1005}
\pps_n&=&
\begin{bmatrix}
0& -1\\
1& 0
\end{bmatrix}
+\frac\a{3n}F_1+\frac\l{(3n)^\a}F_2 \\
&+& \frac{b_{3n}}{(3n)^\a}F_3\lb(1-\frac\a{3n})I_2+\frac\l{(3n)^\a}
\begin{bmatrix}
0& 1\\
1& 0
\end{bmatrix}
\rb +O\lp\frac 1{n^{2\a}}\rp, \nn
\end{eqnarray}
where
$$
F_1=
\begin{bmatrix}
0& 1\\
-2& 0
\end{bmatrix},\quad
F_2=
\begin{bmatrix}
-1& 0\\
0& -2
\end{bmatrix},\quad
F_3=
\begin{bmatrix}
0& 0\\
0& 1
\end{bmatrix}
$$
are constant $2\times 2$ matrices.

The leading term 
$\begin{bmatrix}
0& -1\\
1& \d
\end{bmatrix}$
in \eqref{e1005} has distinct eigenvalues $\d/2\pm\sqrt{\d^2/4-1}, \ \d\not=2,$ and so we can apply the generalized version of Levinson theorem from Janas-Moszy\'nski \cite[Thm. 1.5]{jamo1}, see also the discussion in Sect. \ref{s13}. Indeed, we have $\pps_n=A_n:=V_n+R_n$ where
\begin{eqnarray}\label{e16}
V_n&:=&
\begin{bmatrix}
0& -1\\
1& \d
\end{bmatrix}
+\frac\a{3n} F_1+\frac\l{(3n)^\a} F_2-\frac{b_{3n}\a}{(3n)^{\a+1}} F_3,\\
R_n&:=&
\frac{b_{3n}\l}{(3n)^{2\a}} F_3\cdot 
\begin{bmatrix}
0& 1\\
1& 0
\end{bmatrix}
+O\lp\frac 1{n^{2\a}}\rp, \nn
\end{eqnarray}
and
\begin{enumerate}
\item $\det \pps_n\to 1$, and $\det V_n\to 1$  as $n\to+\infty$,
\item $\{V_n\}\in BV$, 
\item we have
$$
V_\infty:=\lim_{n\to+\infty} V_n=
\begin{bmatrix}
0& -1\\
1& \d
\end{bmatrix}
$$
and $(\tr V_\infty)^2-4\det V_\infty=\d^2-4<0$ provided $0<\d<2$,
\item we see that $R_n=O(1/n^{2\a})\in \ell^1$ as $\a>1/2$.
\end{enumerate}
Hence, by Theorem \ref{t02} the generalized eigenvectors $u^\pm:=u^\pm(\l)=\{u^\pm_n(\l)\}, \ (Ju^\pm)’=\l (u^\pm)’$, are expressed through the eigenvalues $\{\l_\pm(n)\}$ of  the ``3-step'' matrices $\pps_n$. Recalling \eqref{e1005}, we compute $\{\l_\pm(n)\}$ as
\begin{eqnarray}
\l_\pm(n)&=&\l^\pm(\pps_n)=\frac 12\, \tr \pps_n\pm\sqrt{\frac 14(\tr\pps_n)^2-\det\pps_n}\\
&=&\frac{-3\l}{2(3n)^\a}+\frac{b_{3n}}{2(3n)^\a} -\frac{\a\d}{6n}\nn\\
&\pm&\sqrt{\lp \frac\d2-\frac{3\l}{2(3n)^\a}-\frac{\a\d}{6n}\rp^2-\lp 1-\frac 1n\rp^\a}+r_n \nn\\
&=&\frac{\d}{2}-\frac{3\l}{2(3n)^\a} -\frac{\a\d}{6n}\pm i\sqrt{\lp 1-\frac{\d^2}4\rp+\frac{3\d\l}{2(3n)^\a}-\frac{\a(3-\d)}{3n}}+r_n, \nn
\end{eqnarray}
where $\{r_n\}\in \ell^1$. For large $n$, we get $\l_+(n)=\ovl{\l_-(n)}$ and $|\l_\pm(n)|=(1-1/n)^\a, \ 1/2<\a<1$. It follows that the basis solutions $u^\pm(\l), \l\in\br$, do not generate a subordinated solution, and, by Gilbert-Pearson subordinacy  theory, see Sect. \ref{s12}, the spectrum $\s(J)$ of the constructed Jacobi matrix is purely absolutely continuous and $\br_+\subset \s_{ac}(J)=\br$. The proof of the proposition is complete.
\end{proof}

Observe that $0<\d<2$ in the above proof and the upper bound on $\d<2$ is optimal. If we take $\d>2$, then a reasoning similar to Proposition \ref{p1} shows that the spectrum $\s(J)$ is purely discrete on the whole real axis, see also Janas-Naboko \cite{jana3} in this connection.

\subsection{Example 2 to Theorem \ref{t22}, a scalar case}\label{s34}
We give short examples of application of Theorems \ref{t21}, \ref{t22} in this and the next subsection. This subsection is concerned with the scalar-valued Jacobi matrices, and the next one deals with the matrix-valued case.

\begin{proposition}\label{p5}
Let $J=J(\{a_n\},\{b_n\})$ be a scalar-valued Jacobi matrix, $d=1$. Assume that for large $n$
\begin{eqnarray*}
&& a_{2n}\le C_1 n^{\a_1},\quad a_{2n-1}\le C_2n^{\a_2},\\
&& b_{2n}\sim D_1 n^{\b_1}, \quad b_{2n-1}\sim D_2n^{\b_2},
\end{eqnarray*}
where $C_i, D_i>0$, $\a_i,\b_i>0$ for $i=1,2$. Let one of the following conditions holds true:
\begin{enumerate}
\item $2\max\{\a_1,\a_2\}<\b_1+\b_2$,
\item $2\a_1=\b_1+\b_2,\ 2\a_2<\b_1+\b_2$ and $C_1^2<D_1D_2$,
\item $2\a_1=2\a_2=\b_1+\b_2$ and $(C_1+C_2)^2<D_1D_2$.
\end{enumerate}
Then the spectrum $\s(J)$ is purely discrete and it accumulates to $+\infty$ only.
\end{proposition}

It goes without saying that the above condition (2) can be replaced by its symmetric version, 
$2\a_1<\b_1+\b_2,\ 2\a_2=\b_1+\b_2$ and $C_1^2<D_1D_2$.

\begin{proof}
We obviously have $b_{2n-1}\to+\infty,\ b_{2n}=(b_{2n})_+\to +\infty$, so it remains to check condition \eqref{e8} only. In this specific situation it reads as
$$
\limsup_{n\to+\infty}\; b_{2n}^{-1/2}a_{2n-1}b_{2n-1}^{-1/2}+\limsup_{n\to+\infty}\; b_{2n-2}^{-1/2}a_{2n-2}b_{2n-1}^{-1/2})<1,
$$
or
$$
\limsup_{n\to+\infty}\; \frac{C_1(2n)^{\a_1}}{\sqrt{D_1D_2} (2n)^{(\b_1+\b_2)/2}}+
\limsup_{n\to+\infty}\; \frac{C_2(2n)^{\a_2}}{\sqrt{D_1D_2} (2n)^{(\b_1+\b_2)/2}}<1.
$$
Now, it is very easy to finish the proof for cases (1)-(3) from the formulation of the proposition. For the sake of completeness, we give the details for the case (3), the other cases are similar. So,  when $2\a_1=2\a_2=\b_1+\b_2$, the latter limsup’s are
$$
\limsup_{n\to+\infty}\; \frac{C_1 n^{\a_1}}{\sqrt{D_1D_2} n^{(\b_1+\b_2)/2}}
+
\limsup_{n\to+\infty}\; \frac{C_2 n^{\a_2}}{\sqrt{D_1D_2} n^{(\b_1+\b_2)/2}}
=\frac{C_1+C_2}{\sqrt{D_1D_2}},
$$
which is strictly less than one. The proof is finished.
\end{proof}
Note that in case (3) above, when one has $(C_1+C_2)^2=D_1D_2$, the continuous spectrum of $J$ can cover a half-line, see \cite{dana}.

\subsection{Example 3 to Theorem \ref{t22}, a $2\times 2$-matrix case}\label{s35}
In this subsection, we show how Theorem \ref{t22} can be applied to the study of the spectrum of a block Jacobi matrix $J=J(\{A_n\},\{B_n\})$. The computations essentially use its block structure. For simplicity, we take $d=2$.

Let $B_{2n-1}\to+\infty$ in the sense required by Theorem \ref{t22} as $n\to+\infty$, and, for large $j$ 
\begin{equation}\label{e1600}
B_{2n}:=
\lt\{
\begin{array}{lcl}
0_2&,& n=2j,\\
\begin{bmatrix}
0&0\\
0&b_j
\end{bmatrix}&,& n=2j-1.
\end{array}\rt.
\end{equation}
Notice that for large $j$
\begin{equation}\label{e160}
\opp{B^{-1}_{4j}}=0_2, \quad \opp{B^{-1}_{4j-2}}=
\begin{bmatrix}
0&0\\
0& b_j^{-1}
\end{bmatrix},
\end{equation}
and we assume $b_j\to+\infty$ as $j\to+\infty$. 
Furthermore, set
\begin{equation}\label{e17}
A_{2j}=a_{2j}
\begin{bmatrix}
1&0\\
\epp_{2j}&\eta_{2j}
\end{bmatrix},\quad
A_{2j-1}=a_{2j-1}
\begin{bmatrix}
1&0\\
\epp_{2j-1}&\eta_{2j-1}
\end{bmatrix},
\end{equation}
where $a_{2j}, a_{2j-1}>0$ and $\eta_{2j}, \eta_{2j-1}\not =0$, so the inverses to $A_{2j}, A_{2j-1}$ exist.

\begin{proposition}\label{p6}
Let $J=J(\{A_n\},\{B_n\})$ be a $2\times 2$-block Jacobi matrix with block entries $A_n, B_n$ defined in \eqref{e1600}, \eqref{e17}. Let $a_j, \eta_j$ satisfy the assumptions given above. Furthermore, suppose
\begin{eqnarray}
&& \limsup_{j\to+\infty}\; b_j^{-1/2}a_{4j-2}\bi\bi B_{4j-1}^{-1/2}
\begin{bmatrix}
\bar \epp_{4j-2}\\ \bar\eta_{4j-2}
\end{bmatrix}
\bi\bi<1, \label{e181}\\
&& \limsup_{j\to+\infty}\; b_j^{-1/2}a_{4j-3}|\eta_{4j-3}| \bi\bi B_{4j-3}^{-1/2}
\begin{bmatrix}
0\\ 1
\end{bmatrix}
\bi\bi<1. \label{e182}
\end{eqnarray}
Then the $\s(J)\cap (0,+\infty)$ is discrete and it accumulates to $+\infty$ only.
\end{proposition}

\begin{proof}
Once again, this is a straightforward application of Theorem \ref{t22}. Since assumptions (1)-(3) of this theorem are clearly satisfied, we need to check the last condition \eqref{e8} only; hence we are to show that this condition is equivalent to \eqref{e181}, \eqref{e182}.

Since the computations are rather similar and simple, we give details for relation  \eqref{e181} only. Reminding \eqref{e160}, we obtain
\begin{eqnarray*}
&&\opp{B^{-1/2}_{4j}}A_{4j}B^{-1/2}_{4j+1}=0_2,\\
&&\opp{B^{-1/2}_{4j-2}}A_{4j-2}B^{-1/2}_{4j-1}=
b_j^{-1/2}a_{4j-1}
\begin{bmatrix}
0&0\\
\epp_{4j-2}&\eta_{4j-2}
\end{bmatrix} 
B^{-1/2}_{4j-1}.
\end{eqnarray*}
Recalling that $B^{-1/2}_{4j-1}$ is an Hermitian matrix, relation \eqref{e181} follows.
\end{proof}

\medskip\nt
{\bf Acknowledgements.} \  
Both authors are close friends of Leonid Golinskii, and SK is Leonid’s collaborator of a long date. This is our great pleasure to dedicate this article to Leonid’s 65-th anniversary. We wish to Leonid many more years of blossoming scientific activity, as well as many nice mathematical results in the so distinctive “LG soft power” style. 

We would like to thank Prof.~C.~Tretter for helpful remarks on the subject of the paper. SN is partially supported by grants RSF15-11-30007, NCN2013/BST/ 04319 and  by the project SPbGU  N11.42.1071.2016. The research presented in the paper was done during his stay at University of Bordeaux in the framework of ``IdEx International Scholars'' program. SN would like to acknowledge the financial support of the program, as well as the hospitality of the Institute of Mathematics of Bordeaux he was hosted at.

\end{document}